\documentclass[conference]{IEEEtran}
\IEEEoverridecommandlockouts
\usepackage{cite}
\usepackage{hyperref}
\usepackage{amsmath,amsthm,amssymb,amsfonts}
\usepackage{cleveref}
\usepackage{algorithmic}
\usepackage{graphicx}
\usepackage{mathrsfs}
\usepackage[all]{xy}
\usepackage{braket}
\usepackage{bbm}
\usepackage{textcomp}
\usepackage{xcolor}
\newtheorem{theorem}{Theorem}

\newtheorem{lemma}[theorem]{Lemma}
\theoremstyle{definition}
\newtheorem{definition}[theorem]{Definition}

\newcommand{\tr}{\operatorname{tr}}
\newcommand{\norm}[1]{\left\lVert #1 \right\rVert}
\newcommand{\Dtr}{D_{\mathrm{tr}}}
\newcommand{\tracenorm}[1]{\norm{#1}_1}
\newcommand{\tracedist}[2]{\tfrac{1}{2}\,\tracenorm{#1 - #2}}

\def\BibTeX{{\rm B\kern-.05em{\sc i\kern-.025em b}\kern-.08em
    T\kern-.1667em\lower.7ex\hbox{E}\kern-.125emX}}
\begin{document}

\title{Unbounded length minimal synchronizing words for quantum channels over qutrits*\\
\thanks{This work was partially supported by a grant from the Simons Foundation (\#704836
to Bj{\o}rn Kjos-Hanssen).}
}

\author{\IEEEauthorblockN{Bj{\o}rn Kjos-Hanssen}
\IEEEauthorblockA{\textit{Department of Mathematics} \\
\textit{University of Hawai\textquoteleft i at M\=anoa}\\
Honolulu, USA \\
bjoernkh@hawaii.edu}
\and
\IEEEauthorblockN{Swarnalakshmi Lakshmanan}
\IEEEauthorblockA{\textit{Department of Mathematics} \\
\textit{University of Hawai\textquoteleft i at M\=anoa}\\
Honolulu, USA \\
slaks@hawaii.edu}
}

\maketitle

\begin{abstract}
Grudka, Karczewski, Kurzynski, Stempin, W\'ojcik and Wojcik (2025) constructed quantum channels with synchronizing words of length 3 for qutrits. We extend their result to arbitrarily long minimal synchronizing words, providing a contrast to \v{C}ern\'y's conjecture for finite automata.
\end{abstract}


\section{Introduction}
Grudka et al.~\cite{10.1088/1367-2630/ae16ca}.
constructed quantum channels with synchronizing words of length 3 for qutrits. We extend their result to arbitrary long minimal synchronizing words.

\begin{definition}
A deterministic automaton (DA) $M$ consists of a set of states $Q$, an alphabet $\Sigma$,
a start state $q_0$, a set of accepting states $F\subseteq Q$, and a transition function $\delta:Q\times\Sigma\to Q$.
If $Q$ and $\Sigma$ are finite sets then $M$ is a deterministic finite automaton (DFA).
Let $\lambda$ denote the empty word and let $\Sigma^*$ be the set of finite words over $\Sigma$.
The function $\delta$ induces a map $\delta^*: Q \times \Sigma^* \to Q$ by
\[
\delta^*(q,\lambda)=q,\quad \delta^*(q,wa)=\delta(\delta^*(q,w),a)
\]
where $w\in\Sigma^*$ and $a\in\Sigma^*$.
\end{definition}
\begin{definition}
A \emph{synchronizing word} in a DFA over an alphabet $\{A,B\}$
is a word $w\in\{A,B\}^*$ such that
$\delta^*(q,w)=\delta^*(r,w)$ for all states $r$ and $q$.
\end{definition}

Our study is motivated by \v{C}ern\'y's conjecture \cite{MR168429} on the shortest synchronizing words in DFAs.
J\'an \v{C}ern\'y demonstrated that for some DFAs with $q$ many states, the shortest synchronizing word has length $(q-1)^2$.
The longstanding \emph{\v{C}ern\'y's conjecture} asserts that whenever a DFA has a synchronizing word,
it has one of length at most this bound. Our result here shows that if we count only basis states (i.e., dimension) then in a quantum setting \v{C}ern\'y's conjecture fails.
The reason is simply that there is no finite upper bound on the minimal synchronizing word length at all, as a function of dimension.

We prove in \Cref{mainthm} that for each $l\in\mathbb N$ there is a quantum channel of qutrits
(hence having $q=3$ basis states) in which there is a synchronizing word, but none of length at most $l$.
\section{Construction}
Our construction is a modification of one due to Grudka et al.
Let
\[
	A_1 = \begin{pmatrix} 0&0&0\\ 1&0&0\\ 0&0&0\end{pmatrix},\quad
	A_2 = \begin{pmatrix}
	0 & 0 & 0\\
	0 & 0 & -1\\
	0 & 1 & 0
	\end{pmatrix}
\]

and define a Kraus operator $A$ by $A(\rho) = A_1 \rho A_1^* + A_2 \rho A_2^*$.

For each positive integer $n$, let
\[
	B = B_n = \begin{pmatrix}
	\cos\theta & -\sin\theta & 0\\
	\sin\theta & \cos\theta & 0\\
	0 & 0 & 1
	\end{pmatrix}
\]
where $\theta=\pi/(2n)$,
and define the Kraus operator $B$ by $B(\rho)=B\rho B^*$.

Let $e_1,e_2,e_3$ be the standard basis vectors for $\mathbb C^3$, so that for example
\[
\ket{e_2}\bra{e_2}=
\begin{pmatrix}
	0 & 0 & 0\\
	0 & 1 & 0\\
	0 & 0 & 0
\end{pmatrix}
\]

\begin{definition}
For any positive integer $n$, $\mathscr M_n$ is the DA whose states are density matrices over $\mathbb C^3$, whose alphabet is $\{A,B_n\}$, with transition function
\[
\delta(\rho,C)=C(\rho)
\]
for $C\in\{A,B_n\}$.
The countable subautomaton $\mathscr M_n'$ is identical to $\mathscr M_n$ except that its set of states is
\[
Q = \{\delta^*(\ket{e_1}\bra{e_1}, w) ;\, w\in\{A,B_n\}^*\}
\]
i.e., density matrices reachable from $\ket{e_1}\bra{e_1}$ using $A$ and $B_n$.
\end{definition}

Notice that both $A$ and $B$ are quantum channels, in particular trace-preserving and satisfying
\begin{equation}\label{next-best}
\sum A_i^*A_i = \mathbbm{1} = B^*B.    
\end{equation}

If $A$ and $B$ were both injective they could not generate a synchronizing word.
For example, $\mathscr{M}'_1$ in its entirety looks like this:
    \[
    \xymatrix{
    {\ket{e_3}\bra{e_3}}\ar@{<->}[r]_A\ar@(dl,dr)_B &{\ket{e_2}\bra{e_2}}\ar@{<->}@/_/[r]_B &{\ket{e_1}\bra{e_1}}\ar@/_/[l]_A & 
    }
    \]
It is verifiable by inspection that the word $ABA$ is synchronizing, and the non-injectivity of $A$ is crucial for this to happen.
(See \Cref{fig:7} for a finite part of
the infinite, and more complicated, DA $\mathscr{M}_2'$.)
The fact that \eqref{next-best} can be maintained is a surprising next best thing after injectivity.

\section{Main results}

Let a positive integer $l$ be given.
Our strategy is simple:
 If we take the angle $\theta$ very small compared to $l$, then the operator $B$ is close to the identity $I$, hence any word of length $l$ or less is approximately just a power $A^p$.
 As $A$ permutes $\ket{e_2}\bra{e_2}$ and $\ket{e_3}\bra{e_3}$, it follows that there is no synchronizing word of length $l$ or less.

To speak precisely about $B$ being close to $I$ we use the trace distance.

\begin{definition}
The trace distance between $\rho$ and $\sigma$ is given by
\( \Dtr(\rho,\sigma) = \tracedist{\rho}{\sigma} = \tfrac{1}{2}\tr\!\big(\sqrt{(\rho - \sigma)^2}\big) \).
\end{definition}
We recall the well-known \Cref{contracts}.
\begin{lemma}\label{contracts}
For any quantum channel $\Phi$,
\[
\Dtr(\Phi(\rho),\Phi(\sigma))\le
\Dtr(\rho,\sigma).
\]    
\end{lemma}

\begin{lemma}[Hölder inequality for Schatten norms \cite{MR2978290} (5.6)]\label{schatten}
Let $A,B$ be complex (or real) matrices, and let $1 \le p,q,r \le \infty$ satisfy
\[
  \frac{1}{r} \;\le\; \frac{1}{p} + \frac{1}{q}.
\]
Then
\[
  \|AB\|_r \;\le\; \|A\|_p \, \|B\|_q,
\]
where $\|\cdot\|_p$ denotes the Schatten $p$-norm,
\[
  \|X\|_p = 
  \begin{cases}
    \bigl(\operatorname{tr}[(X^* X)^{p/2}]\bigr)^{1/p}, & 1 \le p < \infty, \\[6pt]
    \max_i s_i(X), & p = \infty,
  \end{cases}
\]
and $s_i(X)$ are the singular values of $X$.

In particular, for $p = q = \infty$ and $r = 1$,
\[
  \|A B\|_1 \le \|A\|_\infty \, \|B\|_1.
\]
\end{lemma}

We now obtain a trace-distance bound for a rank-one channel close to the identity.
\begin{lemma}\label{chat}
Let \( B \) be a \( 3\times3 \) real matrix, and define a quantum channel
\[
\Phi(\rho) = B\rho B^{T}.
\]
Assume that every entry of \( B \) differs from the identity by at most \( \epsilon\le 2/9 \):
\[
B = I + E, \qquad |E_{ij}| \le \epsilon.
\]
Then $\Dtr(\rho,\Phi(\rho))\le 4\epsilon$.    
\end{lemma}
\begin{proof}
We would like to bound the trace distance
\[
D_{\mathrm{tr}}(\rho,\Phi(\rho))
   = \tfrac{1}{2}\|\Phi(\rho)-\rho\|_1
   = \tfrac{1}{2}\|B\rho B^{T}-\rho\|_1.
\]

Write
\[
B\rho B^{T}-\rho
   = (B-I)\rho B^{T} + \rho(B^{T}-I).
\]

Applying the triangle inequality and \Cref{schatten}
\[\|XYZ\|_1 \le \|XY\|_1 \|Z\|_{\infty} \le \|X\|_{\infty}\,\|Y\|_1\,\|Z\|_{\infty},\]
we get
\[
\|B\rho B^{T}-\rho\|_1
   \le \|B-I\|_{\infty}\,\|\rho\|_1\,\|B\|_{\infty}
      + \|\rho\|_1\,\|B^{T}-I\|_{\infty}.
\]
Since \(\|\rho\|_1 = \operatorname{tr}\rho = 1\)
and \(\|B^{T}-I\|_{\infty} = \|B-I\|_{\infty}\),
\begin{equation}\label{boxed}
\|B\rho B^{T}-\rho\|_1
   \le \|B-I\|_{\infty}\,(\|B\|_{\infty} + 1).    
\end{equation}

For a \(3\times3\) matrix, the Frobenius norm satisfies
\[\|M\|_{\infty}\le \|M\|_F= \sqrt{\sum_{i,j}|M_{ij}|^2} \le 3\max_{i,j}|M_{ij}|.
\]
Hence, using \(|E_{ij}| \le \epsilon\),
\[
\|B-I\|_{\infty} = \|E\|_{\infty} \le \|E\|_F \le 3\epsilon.
\]
Also,
\[
\|B\|_{\infty} \le \|I\|_{\infty} + \|E\|_{\infty} \le 1 + 3\epsilon.
\]
Substituting into \eqref{boxed} gives
\[
\|B\rho B^{T}-\rho\|_1
   \le (3\epsilon)\bigl(1 + 1 + 3\epsilon\bigr)
   = 3\epsilon(2 + 3\epsilon).
\]
Finally,
\[
D_{\mathrm{tr}}(\rho,\Phi(\rho))
   = \tfrac{1}{2}\|B\rho B^{T}-\rho\|_1
   \le 3\epsilon + \tfrac{9}{2}\epsilon^2\le 4\epsilon.\qedhere
\]
\end{proof}

\begin{lemma}
For our choice of $B=B_n$,
the values $|E_{ij}|$ in \Cref{chat} are
bounded by $|\theta|$.    
\end{lemma}
\begin{proof}
These values are 0,
$\cos\theta-1$, and $\pm\sin\theta$,
and we have $|\sin\theta|\le |\theta|$
and $|\cos\theta-1|\le |\theta|$.    
\end{proof}

\begin{lemma}\label{oct30}
    If $\Dtr(B(\rho),\rho)\le \delta$ for
    all $\rho$ then for all integers $s\ge 0$ and all $\rho$,
    $\Dtr(B^s(\rho),\rho)\le s\cdot\delta$.
\end{lemma}
\begin{proof}
    For $s=0$, $\Dtr(B^0(\rho),\rho)=\Dtr(\rho,\rho)=0$.
    Inductively,
    \begin{align*}
    \Dtr (B^{s+1}(\rho),\rho) &\le \Dtr(B^{s+1}(\rho),B^s(\rho)) + \Dtr(B^s,\rho) \nonumber\\
    &\le \delta + s\cdot \delta = (s+1)\cdot \delta.\qedhere
    \end{align*}
\end{proof}
\begin{lemma}\label{six}
If $\Dtr(B(\rho),\rho)\le \delta$ for
    all $\rho$ then
\[
\Dtr\left(\left(\prod_{i=1}^tA^{a_i}B^{b_i}\right)\rho,\, A^{\sum_i a_i}\rho\right) \le \delta \sum_i b_i.
\]
\end{lemma}
\begin{proof}
By \Cref{contracts} and \Cref{oct30},
\[
\left(\prod_{i=1}^tA^{a_i}B^{b_i}\right)\rho
=A^{a_1}B^{b_1}\left(\prod_{i=2}^tA^{a_i}B^{b_i}\right)\rho
\]
is within $b_1\cdot \delta$ of 
\[
A^{a_1}\left(\prod_{i=2}^tA^{a_i}B^{b_i}\right)\rho,
\]
so by induction 
we are done.    
\end{proof}

\begin{theorem}\label{mainthm}
    For each $l$ there is an $n$
    such that there is no synchronizing
    word in $\mathscr M_n$ or $\mathscr M_n'$ of length at most $l$
    over the alphabet $\{A,B_n\}$.
\end{theorem}
\begin{proof}
Let $\epsilon<1/2$ and $\epsilon'=\frac{\epsilon}{4l}$.
Choose $|\theta| \le \epsilon'$. 
Since $\theta=\pi/(2n)$ this amounts to
choosing $n\ge \pi/(2\epsilon')$.

Let $w$ be a purported synchronizing word of length $l$ or less. We may assume $w$ has length exactly $l$ since for any word $v$, the concatenated word $vw$ is synchronizing as well.
By \Cref{six} with $\rho=\ket{e_2}\bra{e_2}$,
there is a certain $j$ with
\begin{align*}
&\Dtr(w(\ket{e_2}\bra{e_2}), A^j(\ket{e_2}\bra{e_2})) \le \epsilon\\
\text{and}\quad &\Dtr(w(\ket{e_3}\bra{e_3}), A^j(\ket{e_3}\bra{e_3})) \le \epsilon.    
\end{align*}
Since $w$ is synchronizing, 
$w(\ket{e_2}\bra{e_2})
=w(\ket{e_3}\bra{e_3})$.
However, $A$ permutes
$\ket{e_2}\bra{e_2}$ and
$\ket{e_3}\bra{e_3}$, so this is impossible:
\begin{align*}
1 &= \Dtr(\ket{e_2}\bra{e_2},\ket{e_3}\bra{e_3})
\\ &=
\Dtr(A^j(\ket{e_2}\bra{e_2}),A^j(\ket{e_3}\bra{e_3}) \\
&\le \Dtr(A^j(\ket{e_2}\bra{e_2}),w(\ket{e_2}\bra{e_2})\\
&+
\Dtr(w(\ket{e_2}\bra{e_2}),A^j(\ket{e_3}\bra{e_3})\\
&\le 2\epsilon<1.\qedhere
\end{align*}
\end{proof}

\begin{lemma}
    For all positive integers $n$, the 
    word $AB^n_nA$ is synchronizing in $\mathscr M_n$ and in $\mathscr M_n'$
    and maps all states to $\ket{e_2}\bra{e_2}$.
\end{lemma}
\begin{proof}
The matrix $B^n_n$ is independent of $n$ and was called $B$ by Grudka et al.~\cite{10.1088/1367-2630/ae16ca} who proved that $ABA$ is synchronizing; see their equation (73).
\end{proof}

\begin{figure*}
\[
\xymatrix@R=1em@C=1em{
{\begin{pmatrix}
	3 & 3 & 0\\
	3 & 3 & 0\\
	0 & 0 & 2
\end{pmatrix}}\ar[dddddd]_B& & & & & &{\begin{pmatrix}
	3 & 0 & 0\\
	0 & 0 & 0\\
	0 & 0 & 1
\end{pmatrix}}\ar[llllll]_B\ar@{..>}@/_2.0pc/[dddlll]\\
&{\begin{pmatrix}
	1 & 1 & 0\\
	1 & 1 & 0\\
	0 & 0 & 2
\end{pmatrix}}\ar@/_/[dddddl]_A\ar[dddd]_B&&&&{\begin{pmatrix}
	1 & 0 & 0\\
	0 & 0 & 0\\
	0 & 0 & 1
\end{pmatrix}}\ar[llll]_B\ar@{..>}@/_2.0pc/[ddll]_A\\
&&&{\begin{pmatrix}
	1 & 1 & 0\\
	1 & 1 & 0\\
	0 & 0 & 0
\end{pmatrix}}\ar@/_2.0pc/[dddll]_A\ar[d]_B&
{\ket{e_1}\bra{e_1}}
\ar[l]_B \ar[dl]_A &  &\\
&&&
{\ket{e_2}\bra{e_2}}
\ar@{<->}[dl]_A\ar[r]_B & {\begin{pmatrix}
	1 & -1 & 0\\
	-1 & 1 & 0\\
	0 & 0 & 0
\end{pmatrix}}\ar[u]^B\ar@/^2.0pc/[ddlll]_A & \\
&&{\ket{e_3}\bra{e_3}}
& \\
&{\begin{pmatrix}
	0 & 0 & 0\\
	0 & 1 & 0\\
	0 & 0 & 1
\end{pmatrix}}\ar[rrrr]_B&&&&{\begin{pmatrix}
	1 & -1 & 0\\
	-1 & 1 & 0\\
	0 & 0 & 2
\end{pmatrix}}\ar[uuuu]_B\ar@/^/[dlllll]_A\\
{\begin{pmatrix}
	0 & 0 & 0\\
	0 & 3 & 0\\
	0 & 0 & 1
\end{pmatrix}}\ar@{<->}[dd]_A\ar[rrrrrr]_B
& &  & & & &{\begin{pmatrix}
	3 & -3 & 0\\
	-3 & 3 & 0\\
	0 & 0 & 2
\end{pmatrix}}\ar[uuuuuu]_B
\\
\\
{\begin{pmatrix}
	0 & 0 & 0\\
	0 & 1 & 0\\
	0 & 0 & 3
\end{pmatrix}}
}\\
\]
\caption{A finite part of $\mathscr{M}_2'$.
    Self-loops are not shown and
each density matrix should be normalized to have trace 1.
}\label{fig:7}
\end{figure*}
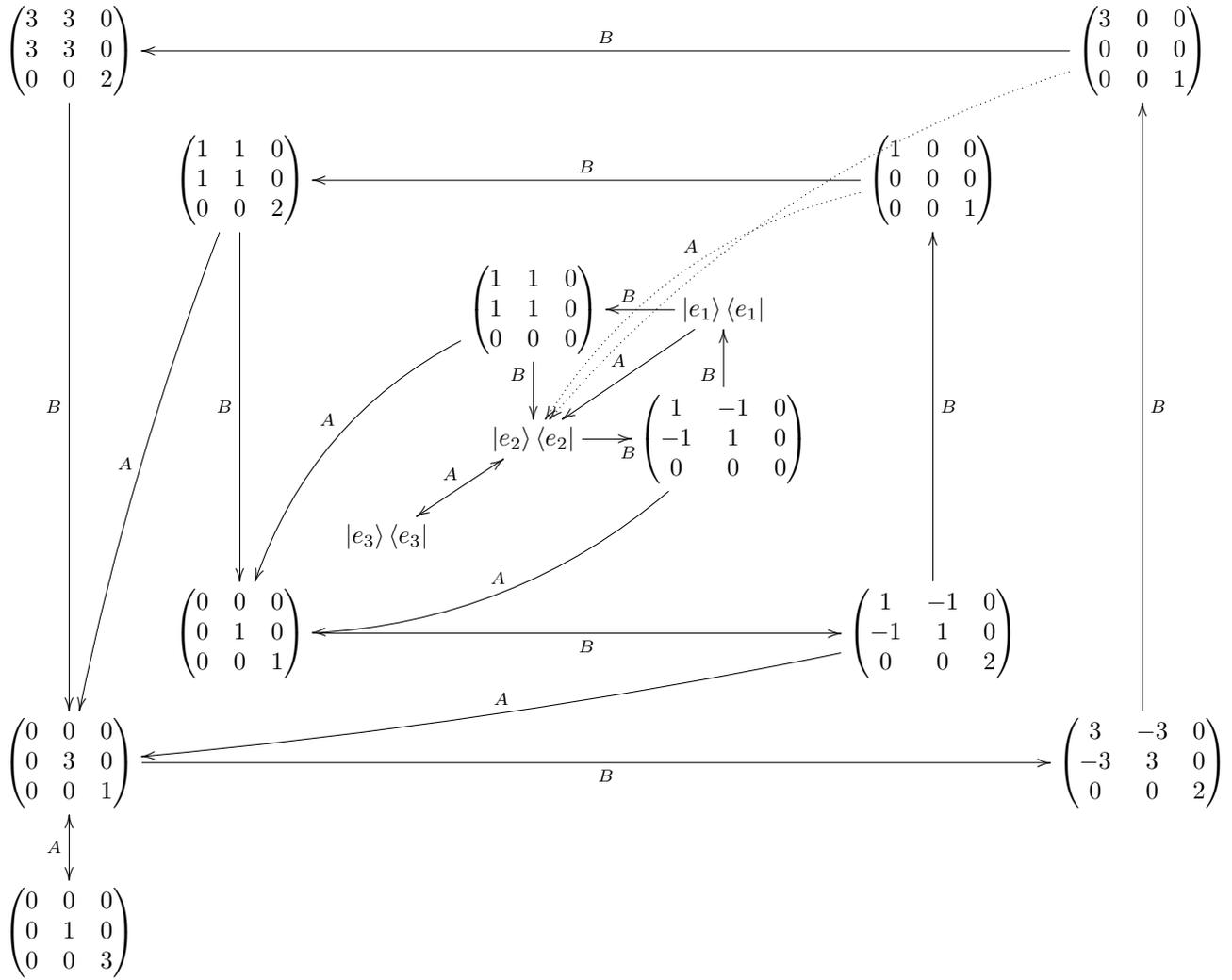

\section{Future work}\label{future}
We mention a couple of questions.
\begin{enumerate}
    \item Do similar results hold for qubits in place of qutrits?
    \item Can we ensure that both letters are given by unital operators, i.e., operators $C$ satisfying $\sum_i C_iC_i^* = \mathbbm{1}$?
\end{enumerate}

\bibliographystyle{IEEEtran}
\bibliography{sample}

\end{document}